\documentclass[letterpaper, 10 pt, conference]{ieeeconf}
\IEEEoverridecommandlockouts
\usepackage{amsmath,amssymb}

\usepackage{amsthm}
\usepackage{mathtools}
\usepackage{graphicx}
\usepackage{bm}
\usepackage{subfig}
\usepackage{algorithm}
\usepackage{algpseudocode}
\usepackage{tabularx}  

\newtheorem{lemma}{Lemma}
\newtheorem{proposition}{Proposition}
\newtheorem{theorem}{Theorem}
\newtheorem{assumption}{Assumption}
\newtheorem{remark}{Remark}
{}

\title{Fast and Scalable Multi-Agent Distribution Matching via Partitioned Optimal Transport}

\author{Kooktae Lee$^{1}$ and Ruchika Singh$^{1}$
	\thanks{*This work was supported by NSF CAREER Grant CMMI-DCSD-2638508.}
	\thanks{Kooktae Lee and $^{1}$Ruchika Singh are with the Department of Mechanical and Aerospace Engineering, Texas Tech University, Lubbock, TX 79409, USA, email: kooktae.lee@ttu.edu, ruchisin@ttu.edu.} 
}

\begin{document}

\maketitle
\thispagestyle{empty}
\pagestyle{empty}

\begin{abstract}
This paper presents a scalable optimal-transport-based framework for
terminal distribution matching in multi-agent systems. While optimal
transport provides a natural way to measure distributional mismatch and
assign agents to a desired spatial distribution, global discrete
transport can become computationally expensive for large-scale systems.
We address this bottleneck by partitioning agents and target samples into
spatially corresponding blocks and solving smaller local transport
problems. Under a mass-balance condition, the resulting restricted
coupling remains feasible for the global problem and provides an upper
bound on the Wasserstein cost. The local assignments generate target
locations for finite-horizon agent control, applicable to both linear and
nonlinear dynamics. By alternating local assignment and control, we
establish a cycle-to-cycle descent guarantee for the resulting transport
surrogate. The proposed framework therefore enables scalable terminal distribution matching while retaining a rigorous connection to the Wasserstein objective. The technical soundness of the proposed results is validated through simulations.
\end{abstract}


\section{Introduction}

Coordinating a group of autonomous agents to achieve a desired spatial
configuration is a fundamental problem in multi-agent systems. In many
applications, the objective is not to drive each agent to a predetermined
location, but to deploy agents across a region according to the spatial
demand of the task. This type of spatial deployment arises in multi-robot
coverage, surveillance, exploration, and environmental sampling, where
coverage and sensing objectives can guide the spatial placement of
multiple agents \cite{cortes2004,inoue2021,bandyopadhyay2014}.
A number of approaches have therefore formulated multi-agent coordination
in terms of spatial distributions, including density-driven optimal
control (D$^2$OC) for coverage under nonuniform reference densities
\cite{seo2025d2oc,seo2026tcst,lee2026automatica}. These approaches
primarily characterize desired spatial behavior through time-averaged
agent trajectories. In this paper, we instead formulate the deployment
objective as a \textit{terminal distribution matching} problem, where the agent
configuration at a prescribed terminal time is matched to a given target
distribution.

A closely related terminal-distribution formulation was developed by
\cite{lee2026tac}, where decentralized multi-agent distribution matching
was achieved through an optimal-transport-based Wasserstein objective.
Optimal transport (OT) \cite{villani2009optimal} provides a natural
framework for this problem by measuring the mismatch between the current
and desired distributions \cite{kantorovich1942,gabriel2019computational} and determining how agents should be
associated with the target distribution. Thus, OT provides both a measure
of distributional mismatch and an assignment mechanism connecting
individual agents to the desired spatial distribution. 
Additional literature on OT-based distribution steering and density control can be found in \cite{chen2021ot,chen2024density}.

The computational cost of this assignment grows with the numbers of
agents and target samples, since discrete OT considers associations
between agents and target samples \cite{gabriel2019computational}. This becomes
particularly challenging for large-scale multi-agent systems, motivating
scalable and distributed OT methods \cite{krishnan2025}. Spatial
structure has also been exploited in semi-discrete OT to obtain
tractable transport partitions \cite{hartmann2020}. These approaches
motivate exploiting spatial structure to reduce the computational burden
of OT while retaining its role in assigning agents to a desired spatial
distribution.

This paper addresses the computational bottleneck of large-scale OT
by exploiting spatial structure in the assignment problem. Our prior
work \cite{lee2026tac} developed both a centralized sequential-assignment framework
and a memory-based decentralized variant for terminal distribution
matching. The present paper builds on the centralized formulation of
\cite{lee2026tac}, which is the natural choice for applications requiring structured,
globally coordinated deployment. The decentralized variant of \cite{lee2026tac} instead
targets settings where such global coordination is unavailable, and is
not the focus here. Even in the centralized setting, however, the
assignment step of \cite{lee2026tac} sequentially couples $M$ agents with $N$ target
samples, and this coupling becomes computationally expensive as the
number of target samples used to represent the desired distribution
grows. To overcome this limitation, this work introduces a spatially
partitioned computational architecture.

The contributions of this paper are thus threefold:
(i) a spatially partitioned OT formulation that strictly preserves global feasibility via block mass-balance constraints;
(ii) an integration of these localized assignments with the barycentric control and target-generation framework; and
(iii) a rigorous cycle-to-cycle descent guarantee for the transport surrogate, which maintains a direct upper bound on the Wasserstein distance.

\section{Preliminaries and Problem Formulation}
\label{sec:prelim}

This section introduces the multi-agent dynamics and terminal
distribution matching problem, followed by the discrete optimal transport
formulation used to measure distribution mismatch. We then summarize the
sequential assignment--control framework of \cite{lee2026tac}, which
serves as the basis for the computational reduction developed in this
work.

\subsection{Terminal Distribution Matching}

Consider a swarm of $M$ agents, where the state of agent $i$ at discrete time $k$ is
$x_i(k)\in\mathbb{R}^d$. Each agent evolves according to either linear
dynamics
\begin{equation}
x_i(k+1)
=
Ax_i(k)+Bu_i(k),
\label{eq:lti}
\end{equation}
or control-affine nonlinear dynamics
\begin{equation}
x_i(k+1)
=
f(x_i(k))+g(x_i(k))u_i(k).
\label{eq:nonlinear}
\end{equation}
Let $N$ target samples $y_j\in\mathbb{R}^d$ represent the desired
spatial distribution. The empirical agent distribution at time $k$ and the target distribution, respectively, are defined by
\begin{equation}
\mu_k
:=
\frac{1}{M}\sum_{i=1}^{M}\delta_{x_i(k)}, \quad
\nu
:=
\frac{1}{N}\sum_{j=1}^{N}\delta_{y_j},
\end{equation}
where $\delta_z$ denotes the Dirac measure concentrated at $z$.
Each agent carries a mass of $1/M$, while each target sample carries
a mass of $1/N$.

Given a terminal time $T$, the objective is to steer the terminal
empirical distribution
$
\mu_T
=
\frac{1}{M}\sum_{i=1}^{M}\delta_{x_i(T)}
$
toward the prescribed target distribution $\nu$. 

\subsection{Discrete Optimal Transport}

The distribution mismatch between $\mu$ and $\nu$ is measured using
the squared 2-Wasserstein distance. The set of transport couplings is
\begin{equation}
\Pi(\mu,\nu)
=
\left\{
\pi\in\mathbb{R}_+^{M\times N}:
\sum_{j=1}^{N}\pi_{ij}=\frac{1}{M},
\;
\sum_{i=1}^{M}\pi_{ij}=\frac{1}{N}
\right\},
\label{eq:transport_set}
\end{equation}
where the transport plan $\pi_{ij}\geq0$ denotes the amount of mass
transported from agent $i$ to target sample $j$. The prescribed row and
column marginals preserve the total mass since they imply
$\sum_{i=1}^{M}\sum_{j=1}^{N}\pi_{ij}=1$.

The squared 2-Wasserstein distance between the empirical distributions
is then
\begin{equation}
\mathcal{W}_2^2(\mu,\nu)
=
\min_{\pi\in\Pi(\mu,\nu)}
\sum_{i=1}^{M}\sum_{j=1}^{N}
\pi_{ij}\|x_i-y_j\|^2.
\label{eq:global_ot}
\end{equation}

Note that for any feasible coupling $\pi\in\Pi(\mu,\nu)$, we have
\begin{equation}
\mathcal{W}_2^2(\mu,\nu)
\leq
\sum_{i=1}^{M}\sum_{j=1}^{N}
\pi_{ij}\|x_i-y_j\|^2.
\label{eq:ot_upper}
\end{equation}
Thus, the transport cost of a feasible coupling provides an upper bound
on the Wasserstein distance without requiring the coupling to be
globally optimal.

\subsection{Terminal Distribution Matching: Sequential Assignment and Control}

The terminal distribution matching framework of \cite{lee2026tac}
addresses the distribution-level objective through an alternating
assignment--control procedure. Rather than directly optimizing the
terminal states of all agents jointly, which is computationally intractable due to high non-convexity and high dimensionality, target mass is assigned to the
agents sequentially. Starting from the prescribed target sample
weights, each agent receives transport mass according to its current
assignment cost, and the assigned target samples determine a barycentric
terminal reference
$
b_i
=
\frac{\sum_{j=1}^{N}\pi_{ij}y_j}
{\sum_{j=1}^{N}\pi_{ij}},
$
where the sequential allocation of target samples ensures that assigned weights are appropriately decremented, preventing multiple agents from being assigned to the same target samples.

The agents are then controlled toward these barycentric targets over a
finite horizon. After the control horizon is completed, the target
sample weights are reinitialized to their prescribed values and the
sequential assignment--control procedure is repeated using the updated
agent states.

The resulting procedure converts the distribution-level matching
problem into a sequence of agent-level finite-horizon control problems.
However, the assignment step still sequentially couples $M$ agents with $N$
target samples, which can become computationally expensive as the
problem size increases. This motivates the partitioned formulation
developed in the next section, where the transport coupling is
restricted to corresponding spatial blocks. Under a suitable mass-balance
condition, the resulting problem decomposes into smaller local
transport problems while preserving feasibility with respect to the
global transport formulation.

\section{Partitioned Optimal Transport}
\label{sec:partition}

\subsection{Spatial Partition and Mass Balance}

Suppose that the $M$ agents and $N$ target samples are partitioned into
$K$ corresponding spatial blocks. The blocks are chosen so that agents
and target samples located in the same spatial region are grouped
together. Let
\begin{equation}
\mathcal{X}
=
\bigcup_{r=1}^{K}\mathcal{X}^{(r)},
\qquad
\mathcal{Y}
=
\bigcup_{r=1}^{K}\mathcal{Y}^{(r)},
\end{equation}
where the blocks are pairwise disjoint. Denote the numbers of agents and
target samples in block $r$ by
\begin{equation}
M_r = |\mathcal{X}^{(r)}|,
\qquad
N_r = |\mathcal{Y}^{(r)}|.
\end{equation}
We impose the following mass-balance condition on corresponding blocks:

\begin{assumption}[Block Mass Balance]
\label{ass:balance}
For every block $r$,
\begin{equation}
\frac{M_r}{M}
=
\frac{N_r}{N}
=:q_r.
\end{equation}
\end{assumption}

Thus, a block containing a fraction $q_r$ of the agents is paired with
a corresponding target block containing the same fraction of the target
samples. Importantly, the corresponding blocks need not contain the same
number of agents and target samples. Only their normalized empirical
masses must agree. For example, if $M=30$ and $N=600$, a block containing
10 agents may be paired with a block containing 200 target samples,
since both blocks contain one-third of their respective populations.
The blocks may be spatially defined, for example, by clustering or a
fixed spatial partition, but no particular partitioning procedure is
required by the analysis. The key requirement is the mass-balance
condition, which ensures that each agent block can be transported
entirely to its corresponding target block.

\subsection{Block-Restricted Coupling and Feasibility}

Define the block-restricted transport set
\begin{equation}
\Pi_{\mathrm{blk}}(\mu,\nu)
=
\left\{
\pi\in\Pi(\mu,\nu) \;\middle|\;
\begin{aligned}
&\pi_{ij}=0 \quad \text{if } x_i\in\mathcal{X}^{(r)}, \\
&\quad y_j\in\mathcal{Y}^{(m)} \text{ with } r\neq m
\end{aligned}
\right\}.
\label{eq:block_coupling}
\end{equation}

The corresponding restricted OT problem is
\begin{equation}
\pi_{\mathrm{blk}}^\star
\in
\arg\min_{\pi\in\Pi_{\mathrm{blk}}(\mu,\nu)}
\sum_{i=1}^{M}\sum_{j=1}^{N}
\pi_{ij}\|x_i-y_j\|^2.
\label{eq:restricted_ot}
\end{equation}

The next result shows that this restriction costs nothing in terms of
feasibility. Any block-diagonal coupling constructed under
Assumption~\ref{ass:balance} is automatically a valid coupling for the
unrestricted problem~\eqref{eq:global_ot}.

\begin{proposition}[Nonemptiness of the  Coupling Set]
\label{prop:feasible}
Under Assumption~\ref{ass:balance}, the set $\Pi_{\mathrm{blk}}(\mu,\nu)$
defined in~\eqref{eq:block_coupling} is nonempty. Since
$\Pi_{\mathrm{blk}}(\mu,\nu) \subseteq \Pi(\mu,\nu)$ by construction,
every block-restricted coupling is therefore a feasible coupling for the
original OT problem, and consequently provides an upper bound on the
squared Wasserstein distance.
\end{proposition}

\begin{proof}
The inclusion $\Pi_{\mathrm{blk}}(\mu,\nu)\subseteq\Pi(\mu,\nu)$ holds
immediately from~\eqref{eq:block_coupling}, which restricts
$\Pi(\mu,\nu)$ by an additional sparsity condition. It remains to show
$\Pi_{\mathrm{blk}}(\mu,\nu)\neq\emptyset$.

Fix a block $r$. The total probability mass carried by the agent
samples in this block is
$
\sum_{x_i\in\mathcal{X}^{(r)}}\frac{1}{M}
=
\frac{M_r}{M}.
$
Similarly, the total target mass in the corresponding block is
$
\sum_{y_j\in\mathcal{Y}^{(r)}}\frac{1}{N}
=
\frac{N_r}{N}.
$
By Assumption~\ref{ass:balance}, these quantities are equal.

Therefore, a coupling between the normalized empirical measures
supported on $\mathcal{X}^{(r)}$ and $\mathcal{Y}^{(r)}$ exists.
After normalization by the common block mass, any standard finite
discrete coupling can be constructed between the two finite uniform
measures. Multiplying it by the common block mass gives a local coupling
whose row sums are $1/M$ and whose column sums are $1/N$.

Perform this construction independently for every block and set all
cross-block entries to zero. The resulting matrix $\pi$ is nonnegative
and satisfies $\sum_j\pi_{ij}=\frac1M$ for every agent $i$ and
$\sum_i\pi_{ij}=\frac1N$ for every target $j$, matching the marginal
conditions in~\eqref{eq:transport_set}. Moreover, by construction
$\pi_{ij}=0$ whenever $x_i\in\mathcal{X}^{(r)}$ and
$y_j\in\mathcal{Y}^{(m)}$ with $r\neq m$, which is exactly the block
restriction in~\eqref{eq:block_coupling}. Hence, $\pi$ is a feasible
global coupling satisfying the prescribed block restriction. In particular, such a coupling exists, thus $\Pi_{\mathrm{blk}}(\mu,\nu)$ is
nonempty.

Finally, since a coupling satisfying the block restriction
of~\eqref{eq:block_coupling} is, by definition, itself a coupling in
$\Pi(\mu,\nu)$, and the Wasserstein cost is the minimum over
$\Pi(\mu,\nu)$, any global coupling $\pi$ satisfying the block
restriction gives
\[
\mathcal{W}_2^2(\mu,\nu)
\leq
\sum_{i,j}\pi_{ij}\|x_i-y_j\|^2.
\vspace{-.35in}
\]
\end{proof}

\subsection{Barycentric Targets}
Let
$\pi\in\Pi(\mu,\nu)$ be the global transport plan obtained by assembling
the local couplings from the corresponding blocks. Since its row marginal
is uniform,
$
\omega_i
:=
\sum_{j=1}^{N}\pi_{ij}
=
\frac1M.
$

The barycentric target associated with agent $i$ is
\begin{equation}
b_i
:=
\frac{1}{\omega_i}
\sum_{j=1}^{N}\pi_{ij}y_j
=
M\sum_{j=1}^{N}\pi_{ij}y_j.
\label{eq:barycenter}
\end{equation}

Because $\pi_{ij}\geq0$ and
$\sum_j\pi_{ij}=1/M$,
the coefficients $M\pi_{ij}$ form a probability vector. As a result,
\begin{equation}
b_i
\in
\operatorname{conv}
\left\{
y_j:\pi_{ij}>0
\right\}.
\label{eq:bary_conv}
\end{equation}
Every agent is therefore assigned a single reachable point rather than a
distribution of targets. The barycenter also admits a natural
variational characterization.

\begin{proposition}
\label{prop:barycenter}
For every agent $i$, the point $b_i$ in~\eqref{eq:barycenter} is the
unique minimizer of
\begin{equation}
F_i(z)
=
\sum_{j=1}^{N}\pi_{ij}\|z-y_j\|^2
\end{equation}
over $z\in\mathbb{R}^d$.
\end{proposition}

\begin{proof}
Differentiating gives
$
\nabla F_i(z)
=
2\sum_j\pi_{ij}(z-y_j).
$
The first-order optimality condition is therefore
\[
\bigg(\sum_j\pi_{ij}\bigg)z
=
\sum_j\pi_{ij}y_j.
\]
Since
$
\sum_j\pi_{ij}=\frac1M>0,
$
the unique stationary point is
\[
z
=
M\sum_j\pi_{ij}y_j
=
b_i.
\]
The Hessian is
$
\nabla^2F_i(z)
=
2\bigg(\sum_j\pi_{ij}\bigg)I
=
\frac{2}{M}I\succ0,
$
thus this stationary point is the unique global minimizer.
\end{proof}

\subsection{Barycentric Control Interpretation}

The control step requires a state-dependent objective that can be analyzed for a fixed transport assignment. Although the Wasserstein distance is the ultimate distribution-matching objective, its definition involves a minimization over the transport coupling, which generally depends on the agent states. Directly analyzing this nested optimization is therefore challenging for deriving agent-level control laws.

To avoid the computational intractability, we freeze the transport coupling computed at the beginning of each control cycle. This design choice decouples the assignment from the ongoing motion, allowing us to treat the assigned targets as stationary references during the finite-horizon control phase. For any
feasible coupling $\pi\in\Pi(\mu,\nu)$, its associated transport cost is
$
\sum_{i=1}^{M}\sum_{j=1}^{N}
\pi_{ij}\|x_i-y_j\|^2.
$
Because $\pi$ is itself a feasible coupling for the Wasserstein problem,
this fixed-coupling cost provides an upper bound on the Wasserstein
distance:
\[
\mathcal{W}_2^2(\mu,\nu)
\leq
\sum_{i=1}^{M}\sum_{j=1}^{N}
\pi_{ij}\|x_i-y_j\|^2.
\]
Thus, reducing the fixed-coupling transport cost reduces a certified
upper bound on the Wasserstein objective.

The connection between this fixed-coupling transport cost and
barycentric target tracking has already been established in our
previous work \cite{lee2026tac}. In particular, the weighted quadratic
transport cost can be decomposed into a barycentric tracking term and a
coupling-dependent residual. We restate this identity below in the
notation used in this paper, as it provides the basis for the
finite-horizon control analysis that follows.

Defining the fixed-coupling transport surrogate cost
\begin{equation}
\Psi_\pi(x)
:=
\sum_{i=1}^{M}\sum_{j=1}^{N}
\pi_{ij}\|x_i-y_j\|^2,
\label{eq:psi}
\end{equation}
the following result holds.

\begin{proposition}[Barycentric Decomposition]
\label{prop:decomposition}
For any feasible coupling $\pi\in\Pi(\mu,\nu)$, define the barycentric
target of agent $i$ as in \eqref{eq:barycenter}.

Then, the fixed-coupling transport cost satisfies
\begin{equation}
\Psi_\pi(x)
=
\frac1M
\sum_{i=1}^{M}
\|x_i-b_i\|^2
+
V_\pi,
\label{eq:exact_decomposition}
\end{equation}
where
\begin{equation}
V_\pi
=
\sum_{i=1}^{M}
\sum_{j=1}^{N}
\pi_{ij}\|y_j-b_i\|^2
\geq0.
\label{eq:variance}
\end{equation}
Moreover, $V_\pi$ is independent of the agent states.
\end{proposition}

The proposition follows directly from the weighted quadratic
barycentric identity established in \cite{lee2026tac}. In particular,
the cross term vanishes by the definition of $b_i$, while the row
marginal condition
$\sum_j\pi_{ij}=1/M$ yields the coefficient $1/M$ in
\eqref{eq:exact_decomposition}.

The residual $V_\pi$ represents the weighted dispersion of the target
samples associated with each agent under the coupling. Consequently,
driving the agents to their barycentric targets does not generally make
the transport cost zero. If every row of $\pi$ transports all of its
mass to a single target, then $V_\pi=0$. In particular, for a one-to-one
permutation coupling, exact barycentric reaching is equivalent to exact
matching of the corresponding target samples.

Because $V_\pi$ is independent of the agent states for a fixed coupling,
minimizing
$
\frac1M\sum_{i=1}^{M}\|x_i-b_i\|^2
$
reduces $\Psi_\pi(x)$ by exactly the same amount. Thus, the transport
assignment determines independent barycentric targets, after which the
control problem can be formulated as a collection of agent-level
finite-horizon target-reaching problems. The novel aspect considered
next is the use of the spatially partitioned coupling, which allows
these transport assignments to be computed through smaller local
problems.

\section{Finite-Horizon Barycentric Control}
\label{sec:control}

Given the barycentric targets of Section~\ref{sec:partition}, this section
addresses how each agent reaches, or contracts toward, its assigned
target over a finite control horizon. Linear dynamics admit an exact
closed-form controller while nonlinear dynamics are treated generically,
through a contraction condition that any suitable controller or
synthesis method may satisfy. As these results are standard finite-horizon
control facts, we state them without proof and refer the interested
reader to the full derivation in our related work \cite{lee2026tac}.

\subsection{Linear Dynamics}

Consider the LTI dynamics~\eqref{eq:lti}. Over a horizon of $H$ steps,
\begin{equation}
x_i(k+H)
=
A^H x_i(k)
+
\Phi_H U_i,
\label{eq:horizon}
\end{equation}
where
$
\Phi_H
:=
\begin{bmatrix}
A^{H-1}B & A^{H-2}B & \cdots & B
\end{bmatrix}
$
and
$
U_i
:=
\begin{bmatrix}
u_i(k)^\top &
u_i(k+1)^\top &
\cdots &
u_i(k+H-1)^\top
\end{bmatrix}^{\top}.
$

\begin{assumption}[Finite-Horizon Reachability]
\label{ass:reachability}
The matrix $\Phi_H$ has full row rank.
\end{assumption}
This ensures reachability over horizon $H$, guaranteeing that arbitrary terminal states are achievable for feasible barycentric tracking.

\begin{lemma}(Minimum-Norm Exact Barycentric Control, Lemma 2 in \cite{lee2026tac})
\label{lemma:lti}
Under Assumption~\ref{ass:reachability}, the minimum-$\ell_2$-norm
control sequence satisfying $x_i(k+H)=b_i$ is
\begin{equation}
U_i^\star
=
\Phi_H^\top
(\Phi_H\Phi_H^\top)^{-1}
\left(
b_i-A^Hx_i(k)
\right).
\label{eq:lti_control}
\end{equation}
Consequently, $x_i(k+H)=b_i$ for every agent $i$, implying that the barycentric target
is reached exactly in $H$ steps.
\end{lemma}

\subsection{Nonlinear Dynamics}

For the nonlinear dynamics~\eqref{eq:nonlinear}, we characterize the
control step through a terminal contraction condition with respect to the
assigned barycentric target. Rather than specifying a particular
nonlinear control synthesis method, we assume that the controller used
during each control cycle satisfies the following condition.

\begin{assumption}[Barycentric Terminal Contraction]
\label{ass:contraction}
For some $\rho\in[0,1)$, the controller satisfies
\begin{equation}
\|x_i(k+H)-b_i\|\leq\rho\|x_i(k)-b_i\|
\label{eq:contraction}
\end{equation}
for every agent $i$ during the considered control cycle.
\end{assumption}

This condition can, for example, be established through a finite-horizon
feedback controller or verified for a nonlinear MPC implementation.
Any control method satisfying~\eqref{eq:contraction} can be incorporated
into the proposed assignment--control framework without changing the
descent analysis in the following section. The exact LTI controller in Lemma~\ref{lemma:lti} corresponds to the special case $\rho=0$, yielding $x_i(k+H)=b_i$ at the end of the prescribed control horizon $H$.

\section{Cycle-to-Cycle Descent and Wasserstein Guarantees}
\label{sec:descent}

Building on single-cycle results, this section establishes guarantees for the full alternating assignment-control iteration. Unlike standard receding-horizon control that updates actions at every step, the proposed framework applies the computed control sequence over the entire horizon $H$ prior to re-solving the restricted optimal transport problem and updating barycentric targets. We demonstrate that the resulting surrogate cost $\Psi$ decreases monotonically across cycles, bounds the true Wasserstein distance from above, and specializes to the exact LTI controller of Section~\ref{sec:control}.

\subsection{Control Cycles and the Control-Step Descent}

Let the control cycles be indexed by $\ell = 0, 1, 2, \dots$, where each cycle spans a duration of $H$ steps such that $k_{\ell+1} = k_\ell + H$.

At the beginning of cycle $\ell$, let
$
\pi^\ell
$
be an optimal global transport plan of the restricted OT problem at
$k_\ell$, i.e., $\pi^\ell$ is a feasible global coupling whose nonzero
entries satisfy the prescribed block restriction. Define the associated
barycentric targets by
$
b_i^\ell
=
M\sum_{j=1}^{N}
\pi_{ij}^\ell y_j.
\label{eq:cycle_barycenter}
$

During cycle $\ell$, the coupling $\pi^\ell$ and the barycentric targets
$b_i^\ell$ are fixed. Define the cycle-level surrogate using \eqref{eq:exact_decomposition} by
\begin{equation}
\Psi_\ell
:=
\Psi_{\pi^\ell}(x(k_\ell)).
\label{eq:cycle_psi}
\end{equation}

To analyze reassignment, one additional feasibility condition is needed. The restricted coupling computed at the start of a cycle must remain a
valid candidate coupling at the end of the cycle, so that the
reoptimized coupling of the next cycle can be compared against it.

\begin{assumption}[Persistent Block Membership]
\label{ass:block_persistence}
The block membership of every agent remains unchanged over each control
cycle. Consequently, the block-restricted sparsity pattern defining
$\Pi_{\mathrm{blk}}$ does not change over the cycle, thus the coupling
$\pi^\ell$ computed at the beginning of the cycle remains feasible for
the block-restricted OT problem defined at the new agent positions
$x(k_{\ell+1})$ at the end of the cycle.
\end{assumption}

This assumption is naturally satisfied when partitions are defined by fixed agent assignments or when the control horizon is chosen such that agents do not migrate across block boundaries.

\begin{lemma}[Control-Step Descent]
\label{lem:control_descent}
Under Assumption~\ref{ass:contraction}, the fixed-coupling surrogate
satisfies
\begin{equation}
\begin{aligned}
&\Psi_{\pi^\ell}(x(k_{\ell+1}))
-
\Psi_{\pi^\ell}(x(k_\ell))\\
&\quad\leq
-(1-\rho^2)
\frac1M
\sum_{i=1}^{M}
\|x_i(k_\ell)-b_i^\ell\|^2.
\end{aligned}
\label{eq:control_descent}
\end{equation}
\end{lemma}

\begin{proof}
From \eqref{eq:exact_decomposition},
$
\Psi_{\pi^\ell}(x(k))
=
\frac1M
\sum_i
\|x_i(k)-b_i^\ell\|^2
+
V_{\pi^\ell}.
$
Since the coupling is fixed throughout the cycle, the residual
$V_{\pi^\ell}$ is identical at $k_\ell$ and $k_{\ell+1}$.

Assumption~\ref{ass:contraction} gives, for every $i$,
\[
\|x_i(k_{\ell+1})-b_i^\ell\|^2
\leq
\rho^2
\|x_i(k_\ell)-b_i^\ell\|^2.
\]
Summing over $i$ and multiplying by $1/M$ gives
\[
\frac1M
\sum_i
\|x_i(k_{\ell+1})-b_i^\ell\|^2
\leq
\frac{\rho^2}{M}
\sum_i
\|x_i(k_\ell)-b_i^\ell\|^2.
\]
Subtracting the decomposition at $k_\ell$ from the decomposition at
$k_{\ell+1}$ cancels the residual term and yields
\[
\begin{aligned}
&\Psi_{\pi^\ell}(x(k_{\ell+1}))
-
\Psi_{\pi^\ell}(x(k_\ell))\leq
\frac{\rho^2-1}{M}
\sum_i
\|x_i(k_\ell)-b_i^\ell\|^2,
\end{aligned}
\]
which is~\eqref{eq:control_descent}.
\end{proof}

\subsection{Cycle-to-Cycle Surrogate Descent}

Lemma~\ref{lem:control_descent} only accounts for the control step under
a coupling that is fixed. The following theorem shows that
reoptimizing the coupling at the end of the cycle can only help, so that
the two effects combine into a single monotone descent guarantee for the
full iteration.

\begin{theorem}[Cycle-to-Cycle Surrogate Descent]
\label{thm:cycle_descent}
Suppose Assumptions~\ref{ass:balance},
\ref{ass:contraction}, and~\ref{ass:block_persistence} hold. Let
$\pi^\ell$ be an optimal global transport plan satisfying the block
restriction at the beginning of cycle $\ell$, and let $\pi^{\ell+1}$ be
an optimal global transport plan satisfying the block restriction
computed at the beginning of cycle $\ell+1$. Then, the sequence
$\{\Psi_\ell\}_{\ell\geq0}$ is monotonically nonincreasing, satisfying
\begin{equation}
\Psi_{\ell+1}
\leq
\Psi_\ell
-
(1-\rho^2)
\frac1M
\sum_{i=1}^{M}
\|x_i(k_\ell)-b_i^\ell\|^2.
\label{eq:global_descent}
\end{equation}
\end{theorem}

\begin{proof}
From Lemma~\ref{lem:control_descent},
\begin{equation}
\Psi_{\pi^\ell}(x(k_{\ell+1}))
\leq
\Psi_{\pi^\ell}(x(k_\ell))
-
(1-\rho^2)
\frac1M
\sum_i
\|x_i(k_\ell)-b_i^\ell\|^2.
\label{eq:first_step}
\end{equation}

Consider the reassignment step at time $k_{\ell+1}$. By
Assumption~\ref{ass:block_persistence}, the old global coupling remains
feasible for the block-restricted transport problem at the new state.
On the other hand, $\pi^{\ell+1}$ is defined as an optimizer over the
same restricted feasible set. Therefore,
\[
\sum_{i,j}
\pi_{ij}^{\ell+1}
\|x_i(k_{\ell+1})-y_j\|^2
\leq
\sum_{i,j}
\pi_{ij}^{\ell}
\|x_i(k_{\ell+1})-y_j\|^2.
\]
By the definition in \eqref{eq:cycle_psi},
\begin{equation}
\Psi_{\ell+1}
\leq
\Psi_{\pi^\ell}(x(k_{\ell+1})).
\label{eq:second_step}
\end{equation}

Combining~\eqref{eq:first_step} and~\eqref{eq:second_step} yields
\[
\begin{aligned}
\Psi_{\ell+1}
&\leq
\Psi_{\pi^\ell}(x(k_{\ell+1}))\\
&\leq
\Psi_\ell
-
(1-\rho^2)
\frac1M
\sum_i
\|x_i(k_\ell)-b_i^\ell\|^2,
\end{aligned}
\]
which proves~\eqref{eq:global_descent} and completes the proof.
\end{proof}

\begin{remark}
Since $0 \leq \rho < 1$, the coefficient $1 - \rho^2$ is strictly positive. Consequently, if at least one agent is not positioned at its assigned barycentric target at the beginning of the cycle ($x_i(k_\ell) \neq b_i^\ell$), the decrement term is strictly positive, yielding strict descent ($\Psi_{\ell+1} < \Psi_\ell$).
\end{remark}

\begin{remark}
Theorem~\ref{thm:cycle_descent} is a statement about the complete
assignment-control iteration. The first inequality \eqref{eq:first_step} is due
to barycentric control, while the second \eqref{eq:second_step} follows from reoptimizing the
restricted transport problem. Thus, unlike a fixed-coupling descent
result, the theorem accounts for the reassignment step.

The persistent block-membership condition is essential for this argument.
If agents cross block boundaries between consecutive cycles, the previous
coupling may no longer satisfy the new block sparsity constraint, and
therefore it cannot in general be used as a feasible comparison point for
the new restricted OT problem.
\end{remark}

\subsection{Wasserstein Upper Bound}

While the preceding analysis establishes the monotonic decrease of the transport surrogate $\Psi_\ell$ across control cycles, its ultimate significance lies in how it bounds the true Wasserstein distance. Because every block-restricted transport plan constitutes a valid feasible point within the global transport set, the evaluated surrogate directly serves as an upper bound on the exact Wasserstein metric. Consequently, the cycle-to-cycle descent guarantees established for $\Psi_\ell$ translate into certified performance bounds on the true distribution-matching objective without requiring the computationally expensive evaluation of the global optimal transport map at every step.

\begin{theorem}[Wasserstein Upper Bound]
\label{thm:wasserstein}
Let $\pi^\ell$ be the global coupling used during cycle $\ell$. Under
Assumption~\ref{ass:block_persistence}, for every
$k\in[k_\ell,k_{\ell+1}]$,
\begin{equation}
\mathcal{W}_2^2(\mu_k,\nu)
\leq
\Psi_{\pi^\ell}(x(k)).
\label{eq:wasserstein_bound}
\end{equation}
In particular,
\begin{equation}
\mathcal{W}_2^2(\mu_{k_{\ell+1}},\nu)
\leq
\Psi_{\ell+1}
\leq
\Psi_\ell.
\label{eq:combined_bound}
\end{equation}
\end{theorem}

\begin{proof}
During cycle $\ell$, the coupling $\pi^\ell$ has row sums
$
\sum_j\pi_{ij}^\ell=\frac1M
$
and column sums
$
\sum_i\pi_{ij}^\ell=\frac1N.
$
The agent dynamics change the locations $x_i(k)$ but do not change the
mass associated with each agent. Therefore, the same matrix $\pi^\ell$
has the prescribed marginal masses for the empirical distribution
$\mu_k$ at every time in the cycle.

By Assumption~\ref{ass:block_persistence}, the global coupling $\pi^\ell$
continues to satisfy the required block restriction during the cycle.
Consequently, it remains a feasible global transport coupling, yielding
$
\pi^\ell
\in
\Pi(\mu_k,\nu).
$
Using the definition of the Wasserstein distance,
\[
\mathcal{W}_2^2(\mu_k,\nu)
=
\min_{\pi\in\Pi(\mu_k,\nu)}
\sum_{i,j}\pi_{ij}\|x_i(k)-y_j\|^2,
\]
and evaluating the objective at the feasible candidate $\pi^\ell$ gives
\[
\mathcal{W}_2^2(\mu_k,\nu)
\leq
\sum_{i,j}
\pi_{ij}^\ell
\|x_i(k)-y_j\|^2
=
\Psi_{\pi^\ell}(x(k)).
\]
This proves~\eqref{eq:wasserstein_bound}.

At $k_{\ell+1}$, Theorem~\ref{thm:cycle_descent} gives
$
\Psi_{\ell+1}\leq\Psi_\ell.
$
Since $\pi^{\ell+1}$ is itself a feasible global coupling,
\[
\mathcal{W}_2^2(\mu_{k_{\ell+1}},\nu)
\leq
\Psi_{\ell+1}.
\]
Combining the inequalities proves~\eqref{eq:combined_bound}.
\end{proof}

\begin{remark}
The preceding result does not imply that the true Wasserstein distance is
itself monotonically decreasing at every cycle. The optimal global
coupling can differ from the restricted coupling, and the upper-bound
surrogate can decrease even when the Wasserstein optimum has a different
cycle-to-cycle behavior. The rigorous guarantee is therefore monotonicity
of the restricted surrogate together with an upper-bound certification of
the Wasserstein distance.
\end{remark}

\subsection{The Linear Time-Invariant Case}
For the exact LTI controller of Lemma~\ref{lemma:lti}, the contraction factor is $\rho=0$, and $x_i(k_{\ell+1})=b_i^\ell$ holds for every agent. Consequently, Assumption~\ref{ass:contraction} holds with equality, and the control-step inequality of Lemma~\ref{lem:control_descent} reduces to an equality. Theorem~\ref{thm:cycle_descent} thus simplifies to
\begin{equation}
\Psi_{\ell+1}
\leq
\Psi_\ell
-
\frac1M
\sum_{i=1}^{M}
\|x_i(k_\ell)-b_i^\ell\|^2.
\label{eq:lti_cycle_descent}
\end{equation}

For the control step before reassignment, the inequality is in fact an
equality:
\begin{equation}
\Psi_{\pi^\ell}(x(k_{\ell+1}))
-
\Psi_{\pi^\ell}(x(k_\ell))
=
-\frac1M
\sum_{i=1}^{M}
\|x_i(k_\ell)-b_i^\ell\|^2.
\label{eq:lti_exact_descent}
\end{equation}

The additional decrease in~\eqref{eq:lti_cycle_descent} stems from the subsequent restricted OT reassignment, which updates the coupling plan $\pi$ at the new states to continually pull agents toward optimal target configurations across cycles. 
Thus, in the LTI case, the total cycle improvement consists of an exact barycentric control decrease followed by a further cost reduction from the reassignment, ensuring the monotonic decrease of the sequence $\{\Psi_\ell\}$.

\section{Computational Complexity and Algorithm Summary}
\label{sec:impl}

\subsection{Local Assignment and Computational Complexity}
\label{subsec:computational_complexity}

A principal advantage of the proposed spatial partitioning architecture is its computational scalability. The block-restricted formulation can be implemented by solving one transport problem within each corresponding pair of spatial blocks. For each block $r=1,\ldots,K$, the local transport problem is formulated as
\begin{equation}
\min_{\pi^{(r)}}
\sum_{x_i\in\mathcal{X}^{(r)}}
\sum_{y_j\in\mathcal{Y}^{(r)}}
\pi_{ij}^{(r)}
\|x_i-y_j\|^2,
\label{eq:local_ot}
\end{equation}
subject to
{\small
\begin{equation}
\begin{aligned}
\sum_{y_j\in\mathcal{Y}^{(r)}}
\pi_{ij}^{(r)}
=
\frac1M,
\quad
\sum_{x_i\in\mathcal{X}^{(r)}}
\pi_{ij}^{(r)}
=
\frac1N,
\quad
r=1,\ldots,K.
\label{eq:local_col}
\end{aligned}
\end{equation}
}

At control cycle $\ell$, the global restricted coupling $\pi^\ell$ is obtained by assembling the $K$ local optimal couplings $\{\pi^{(r),\ell}\}_{r=1}^{K}$ along the block diagonal. Thus, the global restricted optimal transport problem is decoupled into $K$ independent local transportation problems, which share no coupling variables and can be executed entirely in parallel at the beginning of each cycle.

For balanced blocks of comparable size, the $r$-th local problem contains approximately $M/K$ agents and $N/K$ target samples. Thus, the proposed partitioning replaces one global transportation problem of size $M\times N$ with $K$ smaller problems of approximately $(M/K)\times(N/K)$ size, substantially reducing the computational burden.

In the special case where $M=N$ and each local problem reduces to a balanced one-to-one assignment, the dispersion term vanishes ($V_\pi=0$), and exact target-reaching by individual agents directly achieves exact empirical distribution matching. For this specific assignment regime, classical algorithms such as the Hungarian algorithm solve each local block of size $M/K$ with $\mathcal{O}((M/K)^3)$ worst-case complexity \cite{kuhn1955hungarian}. Consequently, solving the $K$ local assignment problems yields a total computational complexity of
\[
K\mathcal{O}\!\left(\left(\frac{M}{K}\right)^{\!3}\right)
=
\mathcal{O}\!\left(\frac{M^3}{K^2}\right),
\]
compared with $\mathcal{O}(M^3)$ for the centralized assignment. This architectural decomposition thus achieves an ideal computational reduction of $\mathcal{O}(K^2)$ in the equal-cardinality regime.

\subsection{Algorithm Summary}
Building upon the localized sub-problem formulation via partitioning, the overall framework is implemented as an alternating assignment and control procedure. The spatial partition of the domain is constructed offline and fixed prior to execution, defining persistent agent-block memberships that satisfy Assumption~\ref{ass:block_persistence}.

The complete computational pipeline is summarized in Algorithm~\ref{alg:main}.
\begin{algorithm}[!h]
\small
\caption{Partition-Restricted OT and Barycentric Control}
\label{alg:main}
\begin{algorithmic}[1]
\Require Initial agent states $\{x_i(0)\}_{i=1}^{M}$, target samples
$\{y_j\}_{j=1}^{N}$, mass-balanced partition, horizon $H$, terminal cycle $R$
\For{$\ell=0,1,\ldots,R-1$}
\State Freeze current agent states $\{x_i(k_\ell)\}_{i=1}^{M}$ at cycle $\ell$.
\State \textbf{parfor} $r = 1, \ldots, K$ \textbf{do}
\State \quad Solve the local OT problem~\eqref{eq:local_ot}--\eqref{eq:local_col} using the current states in block $\mathcal{X}^{(r)}$ to obtain $\pi^{(r),\ell}$.
\State \textbf{end parfor}
\State Assemble the global coupling $\pi^\ell = \mathrm{blockdiag}(\pi^{(1),\ell}, \ldots, \pi^{(K),\ell})$.
\State Compute the barycentric targets $b_i^\ell = M\sum_{j=1}^{N}\pi_{ij}^\ell y_j$.
\If{linear dynamics}
\State Apply~\eqref{eq:lti_control} for $H$ steps.
\Else
\State Apply a controller satisfying~\eqref{eq:contraction}.
\EndIf
\EndFor
\end{algorithmic}
\end{algorithm}
At the onset of each control cycle $\ell$, the framework updates the transport coupling by solving independent local sub-problems, computes the resulting barycentric targets, and executes the finite-horizon control law. By Theorem~\ref{thm:cycle_descent}, this iterative execution guarantees the monotonic nonincreasing behavior of the surrogate $\Psi_\ell$.

\section{Simulations}

This section evaluates the distribution-matching behavior, cycle-to-cycle
surrogate descent, and computational benefit of the proposed
partitioned transport formulation.

The simulation uses $M=30$ agents with the LTI dynamics~\eqref{eq:lti},
where
$\footnotesize
A=\begin{bmatrix}0.9&0.1\\0&0.9\end{bmatrix}$,
$\footnotesize
B=\begin{bmatrix}0\\0.1\end{bmatrix},
$
and $N=500$ target samples drawn from a mixture of four anisotropic
2-D Gaussians. The agents initially form a cluster outside the target
support. Each control cycle uses $H=50$ steps, and the
assignment-control procedure is repeated for 20 cycles. For the
partitioned case, the target samples are divided into $K=5$ spatial
blocks, and the corresponding agent blocks are chosen to satisfy the
mass-balance condition in Assumption~\ref{ass:balance}. Agent block
membership is kept fixed across cycles, consistent with
Assumption~\ref{ass:block_persistence}. Both centralized and local
assignments are solved exactly as balanced transportation linear
programs.

\begin{figure}[!t]
    \centering
    \subfloat[]{
    \includegraphics[width=0.49\linewidth]{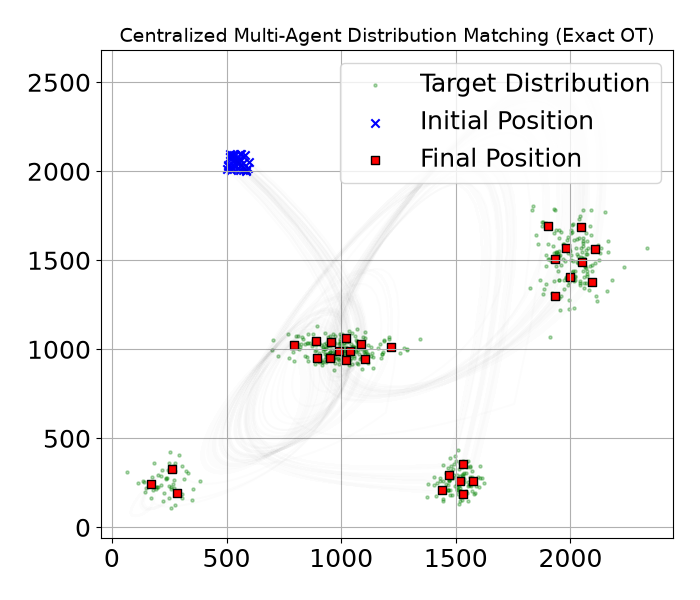}
    }
    \subfloat[]{
    \includegraphics[width=0.49\linewidth]{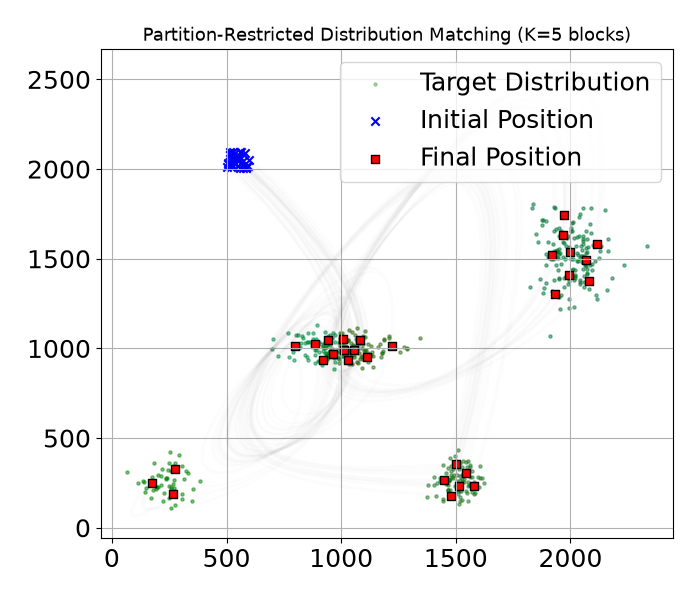}
    }
    \caption{Multi-agent distribution matching trajectories: (a) centralized (non-partitioned) optimal transport and (b) partition-restricted transport with $K=5$ blocks.}
    \label{fig:1}
\end{figure}

Fig.~\ref{fig:1} compares the centralized and partition-restricted
solutions. In both cases, the agents disperse from the initial cluster
and distribute across the four target modes. Proposition~\ref{prop:feasible} guarantees that the
mass-balanced restricted coupling is feasible for the original OT
problem. 
Although the target samples are partitioned by $K$ distinct blocks in the restricted OT method, its overall qualitative behavior for the terminal distribution matching shows a similar performance with respect to the original centralized OT method. 

\begin{figure}[!t]
    \centering
    \includegraphics[width=0.75\linewidth]{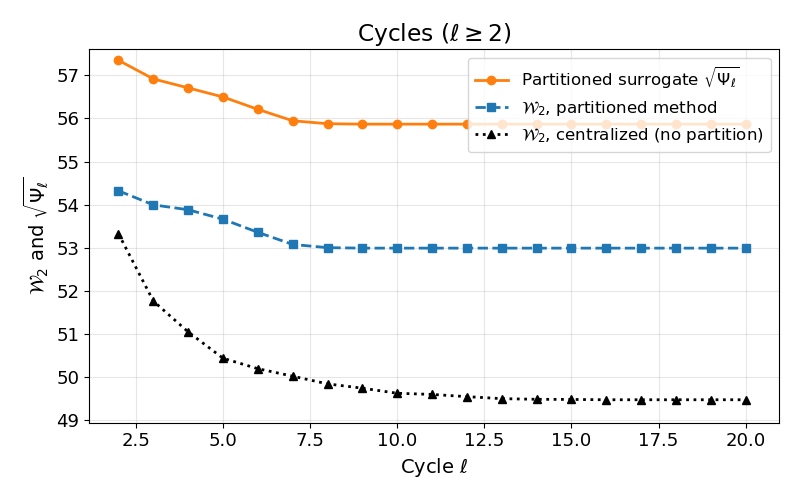}
    \caption{Cycle-to-cycle partitioned surrogate square root $\sqrt{\Psi_\ell}$ and the actual Wasserstein distances $\mathcal{W}_2$ for partitioned and centralized methods over successive control cycles.}
    \label{fig:2}
\end{figure}

Fig.~\ref{fig:2} illustrates the square root of the restricted surrogate, $\sqrt{\Psi_\ell}$ in \eqref{eq:cycle_psi} together with the corresponding Wasserstein distances over successive control cycles. The partitioned surrogate square root $\sqrt{\Psi_\ell}$ decreases monotonically across all cycles, consistent with the cycle-to-cycle descent guaranteed by Theorem~\ref{thm:cycle_descent}. Moreover, the computed Wasserstein distance at each cycle boundary satisfies $\mathcal{W}_2(\mu_{k_{\ell+1}},\nu)\leq\sqrt{\Psi_{\ell+1}}$, as certified by Theorem~\ref{thm:wasserstein}. Thus, Fig.~\ref{fig:2} numerically verifies both the monotone decrease of the partitioned surrogate and its upper-bound relationship with the true Wasserstein metric. While the centralized method achieves a slightly lower Wasserstein distance (better distribution matching) due to its unrestricted transport flexibility, the partitioned approach provides a favorable trade-off, securing substantial computational speedups with only a comparable sacrifice in matching performance.

\begin{figure}[!t]
    \centering
    \includegraphics[width=1\linewidth]{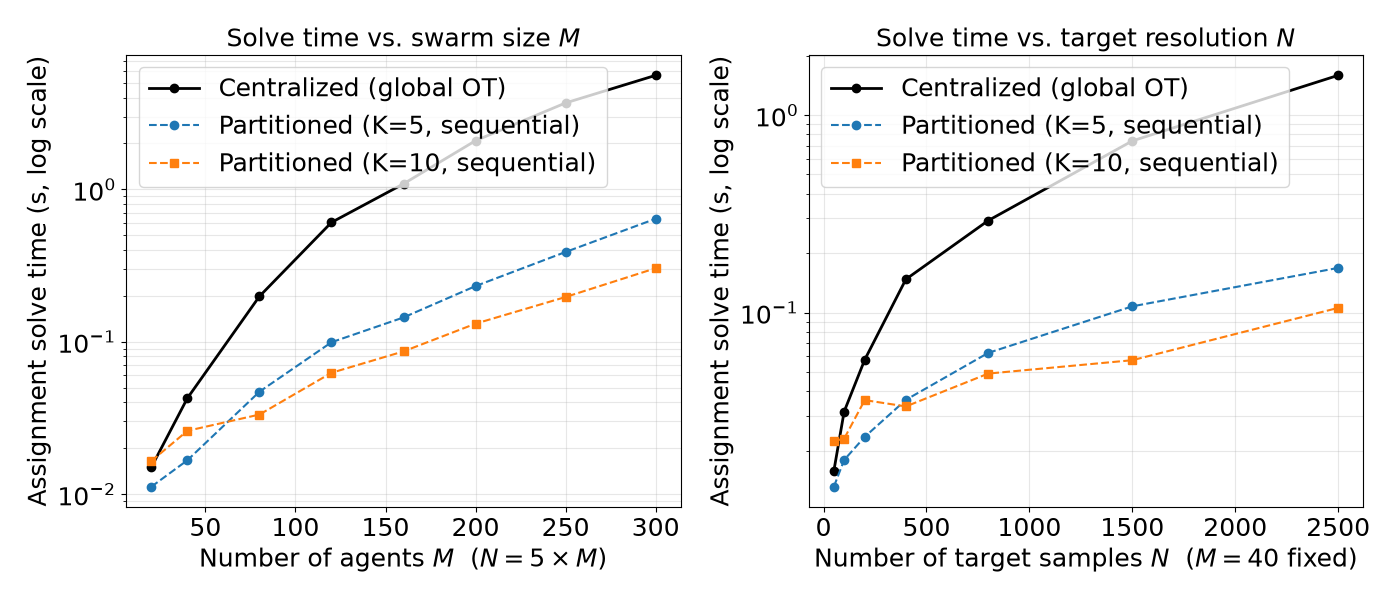}
    \caption{Assignment solve times in logarithmic scale: (left) versus swarm size $M$, and (right) versus target resolution $N$.}
    \label{fig:3}
\end{figure}

Fig.~\ref{fig:3} illustrates the practical benefits of this
restriction. At the largest tested scale ($M=300$, $N=1500$), the
centralized approach requires $5.48\,\mathrm{s}$, while $K=5$ and
$K=10$ achieve solve times of $0.49\,\mathrm{s}$ ($11\times$ speedup)
and $0.21\,\mathrm{s}$ ($26\times$ speedup), respectively. A comparable
trend appears when scaling $N$ at a fixed $M=40$, where $K=10$ operates
roughly $22\times$ faster than the centralized baseline, and in both
sweeps the speedup grows with problem size rather than saturating.
This confirms that partitioning into $K$ blocks reduces each local problem to a
$1/K$ fraction of the agents and target samples, and the resulting
reduction in aggregate solve time is substantial and improves further
as $K$ increases from 5 to 10. The empirical gains are therefore fully
consistent with the computational motivation
for the partitioned formulation.

These simulation results demonstrate the three intended properties of
the proposed framework: feasibility of the mass-balanced restriction,
monotone descent of the restricted surrogate, and reduced assignment
computation through spatial decomposition.

\section{Conclusion}
\label{sec:conclusion}

This paper presented a scalable framework for optimal-transport-based distribution matching by decomposing a global assignment into $K$ spatially restricted local subproblems. By enforcing a block mass-balance condition, the restricted coupling preserves feasibility for the original transport problem, certifying a rigorous Wasserstein upper bound. Combined with finite-horizon barycentric control, the formulation guarantees a monotonic cycle-to-cycle descent of the transport surrogate, which upper-bounds the true Wasserstein metric. Numerical evaluations demonstrated that this spatial restriction significantly enhances computational efficiency, reducing assignment solve times by an order of magnitude or more while maintaining matching performance comparable to global methods. Consequently, the proposed approach renders large-scale distribution matching tractable without compromising theoretical convergence or feasibility guarantees.

\bibliographystyle{ieeetr}
\bibliography{reference}

@article{kuhn1955hungarian,
  title={The Hungarian method for the assignment problem},
  author={Kuhn, Harold W},
  journal={Naval research logistics quarterly},
  volume={2},
  number={1-2},
  pages={83--97},
  year={1955},
  publisher={Wiley Online Library}
}

@book{villani2009optimal,
  title={Optimal transport: old and new},
  author={Villani, C{\'e}dric and others},
  volume={338},
  year={2009},
  publisher={Springer}
}

@article{lee2026tac,
  author  = {Lee, Kooktae},
  title   = {Optimal Transport-Based Decentralized Multi-Agent
             Distribution Matching},
  journal = {IEEE Transactions on Automatic Control},
  volume  = {71},
  number  = {8},
  pages   = {5478--5485},
  year    = {2026},
  doi     = {10.1109/TAC.2026.3668445}
}

@article{seo2025d2oc,
  author  = {Seo, Sungjun and Lee, Kooktae},
  title   = {Density-Driven Optimal Control for Efficient and
             Collaborative Multiagent Nonuniform Coverage},
  journal = {IEEE Transactions on Systems, Man, and Cybernetics:
             Systems},
  volume  = {55},
  number  = {12},
  pages   = {9340--9354},
  year    = {2025},
  doi     = {10.1109/TSMC.2025.3622075}
}

@article{seo2026tcst,
  author  = {Seo, Sungjun and Lee, Kooktae},
  title   = {Density-Driven Multidrone Coordination for Efficient
             Farm Coverage and Management in Smart Agriculture},
  journal = {IEEE Transactions on Control Systems Technology},
  volume  = {34},
  number  = {2},
  pages   = {711--724},
  year    = {2026},
  doi     = {10.1109/TCST.2025.3631091}
}

@article{lee2026automatica,
  author  = {Lee, Kooktae},
  title   = {Density-Driven Optimal Control: Convergence Guarantees
             for Stochastic LTI Multi-Agent Systems},
  journal = {Automatica},
  volume  = {193},
  pages   = {113231},
  year    = {2026},
  doi     = {10.1016/j.automatica.2026.113231}
}

@inproceedings{bandyopadhyay2014,
  author    = {Bandyopadhyay, Saptarshi and Chung, Soon-Jo and Hadaegh, Fred Y.},
  title     = {Probabilistic Swarm Guidance Using Optimal Transport},
  booktitle = {2014 IEEE Conference on Control Applications (CCA)},
  pages     = {498--505},
  year      = {2014},
  doi       = {10.1109/CCA.2014.6981395}
}

@article{hartmann2020,
  author  = {Hartmann, Valentin and Schuhmacher, Dominic},
  title   = {Semi-discrete optimal transport: a solution procedure for the unsquared Euclidean distance case},
  journal = {Mathematical Methods of Operations Research},
  volume  = {92},
  number  = {1},
  pages   = {133--163},
  year    = {2020},
  doi     = {10.1007/s00186-020-00703-z}
}

@article{kantorovich1942,
  title={On the Translocation of Masses.},
  author={Kantorovich, Leonid V},
  journal={Journal of mathematical sciences},
  volume={133},
  number={4},
  pages={1381},
  year={2006}
}

@article{gabriel2019computational,
  title={Computational optimal transport with applications to data sciences},
  author={Gabriel, Peyr{\'e} and Marco, Cuturi},
  journal={Foundations and trends{\textregistered} in machine learning},
  volume={11},
  number={5-6},
  pages={355--607},
  year={2019},
  publisher={Emerald Publishing Limited}
}

@article{cortes2004,
  author  = {Cort{\'e}s, Jorge and Mart{\'i}nez, Sonia and Karatas, Timur and Bullo, Francesco},
  title   = {Coverage Control for Mobile Sensing Networks},
  journal = {IEEE Transactions on Robotics and Automation},
  volume  = {20},
  number  = {2},
  pages   = {243--255},
  year    = {2004},
  doi     = {10.1109/TRA.2004.824698}
}

@article{inoue2021,
  author  = {Inoue, Daisuke and Ito, Yuji and Yoshida, Hiroaki},
  title   = {Optimal Transport-Based Coverage Control for Swarm Robot Systems: Generalization of the Voronoi Tessellation-Based Method},
  journal = {IEEE Control Systems Letters},
  volume  = {5},
  number  = {4},
  pages   = {1483--1488},
  year    = {2021},
  doi     = {10.1109/LCSYS.2020.3039008}
}

@article{chen2021ot,
  author  = {Chen, Yongxin and Georgiou, Tryphon T. and Pavon, Michele},
  title   = {Optimal Transport in Systems and Control},
  journal = {Annual Review of Control, Robotics, and Autonomous Systems},
  volume  = {4},
  pages   = {89--113},
  year    = {2021},
  doi     = {10.1146/annurev-control-070220-100858}
}

@article{chen2024density,
  author  = {Chen, Yongxin},
  title   = {Density Control of Interacting Agent Systems},
  journal = {IEEE Transactions on Automatic Control},
  volume  = {69},
  number  = {1},
  pages   = {246--260},
  year    = {2024},
  doi     = {10.1109/TAC.2023.3271226}
}

@article{krishnan2025,
  author  = {Krishnan, Vishaal and Mart{\'i}nez, Sonia},
  title   = {Distributed Online Optimization for Multi-Agent Optimal Transport},
  journal = {Automatica},
  year    = {2025},
  doi     = {10.1016/j.automatica.2024.111880}
}

\end{document}